\documentclass{conm-p-l}%
\usepackage{amsfonts}
\usepackage{amsmath}
\usepackage{amssymb}
\usepackage{graphicx}%
\usepackage{cite}
\newtheorem{theorem}{Theorem}[section]
\newtheorem{corollary}[theorem]{Corollary}
\newtheorem{lemma}[theorem]{Lemma}
\newtheorem{proposition}[theorem]{Proposition}
\theoremstyle{definition}

\theoremstyle{remark}

\numberwithin{equation}{section}
\copyrightinfo{2010}{enter name of copyright holder}
\begin{document}
\title[Sobolev Mapping Properties of the Scattering Transform]{Sobolev Mapping Properties of the Scattering Transform for the Schr\"{o}dinger Equation}
\author[Hryniv]{Rostyslav O. Hryniv}
\address[Hryniv]{Institute for Applied Problems of Mechanics and Mathematics, 3b
Naukova st., 79601 Lviv, Ukraine \and Department of Mechanics and Mathematics,
Lviv National University, 79602 Lviv, Ukraine \and Institute of Mathematics,
the University of Rzesz\'{o}w, 16\,A Rejtana str., 35-959 Rzesz\'{o}w, Poland}
\author[Mykytyuk]{Yaroslav V. Mykytyuk}
\address[Mykytyuk]{Department of Mechanics and Mathematics, Lviv National University,
79602 Lviv, Ukraine}
\author[Perry]{Peter A. Perry}
\address[Perry]{ Department of Mathematics, University of Kentucky, Lexington,
Kentucky, 40506-0027, U.S.A.}

\thanks{This paper is in final form and no version of it will be submitted for
publication elsewhere.}

\subjclass[2000]{Primary 34L25; Secondary 34L40, 47L10, 81U40}

\date{January 10, 2010}

\keywords{inverse scattering transform, nonlinear Fourier transform, Schr\"{o}dinger equation}

\dedicatory{Dedicated to Mikhail Shubin on the occasion of his sixty-fifth birthday}

\begin{abstract}
We consider the scattering transform for the Schr\"{o}dinger equation with a
singular potential and no bound states. Using the Riccati representation for
real-valued potentials on the line, we obtain invertibility and Lipschitz
continuity of the scattering transform between weighted and Sobolev spaces.
Our approach exploits the connection between scattering theory for the
Schr\"{o}dinger equation and scattering theory for the ZS--AKNS system.
\end{abstract}
\maketitle

\section{Introduction}

The purpose of this paper is to study Sobolev space mapping properties of the
direct and inverse scattering maps for the one-dimensional Schr\"{o}dinger
equation with a potential of low regularity and no bound states. One of our
motivations is to use the scattering maps for the Schr\"{o}dinger equation to
construct and study solutions of the KdV equation on the line with initial
data of low regularity using the inverse scattering method. This paper
presents first steps toward this goal which we will continue in
\cite{HMP:2011}.

In this paper, we will describe a new representation for singular potentials
on the line, the Riccati representation, inspired by the work of Kappeler,
Perry, Shubin, and Topalov \cite{KPST:2005} on the Miura map \cite{Miura:1968}%
. As we will see, the Sobolev mapping properties of the scattering map are
particularly transparent when this representation is used. An analogous
representation for Schr\"{o}dinger operators on the circle appears in the work
of Kappeler and Topalov \cite{KT:2003,KT:2005a,KT:2005b,KT:2006} on well-posedness of the periodic KdV and mKdV equations.

If $q$ is a real-valued distribution on the real line belonging to the space
$H^{-1}(\mathbb{R})$, the Schr\"{o}dinger operator $-d^{2}/dx^{2}+q$ may be
defined as the self-adjoint operator associated to the closure of the
semibounded quadratic form
\begin{equation}
\mathfrak{q}(\varphi)=\int\left\vert \varphi^{\prime}(x)\right\vert
^{2}dx+\left\langle q,\left\vert \varphi\right\vert ^{2}\right\rangle
\label{eq:qf}%
\end{equation}
with domain $C_{0}^{\infty}(\mathbb{R})$ (see Appendix B in \cite{KPST:2005}
and references therein). It is natural to begin by considering such singular
potentials without negative-energy bound states, i.e., distributions $q$ for
which the quadratic form (\ref{eq:qf}) is non-negative. As shown in
\cite{KPST:2005}, such a distribution can be presented in the form
\[
q=u^{\prime}+u^{2}%
\]
where $u\in L_{\mathrm{loc}}^{2}(\mathbb{R})$ is the logarithmic derivative of
a positive solution $y\in H_{\mathrm{loc}}^{1}(\mathbb{R})$ of the zero-energy
Schr\"{o}dinger equation $-y^{\prime\prime}+qy=0$. The function $u$ is called
a \emph{Riccati representative} for the distribution $q$.

There is a one-to-one correspondence between Riccati representatives $u$ and
strictly positive solutions $y$ to the zero-energy Schr\"{o}dinger equation,
normalized so that $y(0)=1$. This latter set consists either of a single point
or a one-parameter family of solutions. Explicitly, $y=\theta y_{-}
+(1-\theta)y_{+}$, where $y_{\pm}$ are the unique, normalized, positive
solutions with the property that
\[
      \int_{0}^{\infty}\frac{ds}{y^{2}_{+}(s)}
    = \int_{-\infty}^{0}\frac{ds}{y^{2}_{-}(s)}
    = \infty
\]
(see \S 5 of \cite{KPST:2005}). If we set $u_{\pm}=\frac{d}{dx}\log y_{\pm}$,
these \textquotedblleft extremal\textquotedblright\ Riccati representatives
$u_{\pm}$ have the property that $v:=u_{-}-u_{+}$ is a nonnegative, H\"{o}lder
continuous function and is either strictly positive, if $u_{+}\neq u_{-}$, or
identically zero, if $u_{+}=u_{-}$.

We can now describe the class of potentials we will study and define the
Riccati representation for such potentials that will play a central role in
our work. Denote by $\mathcal{Q}$ the set of real-valued distributions $q\in
H^{-1}(\mathbb{R})$ with the properties that

(i) the quadratic form (\ref{eq:qf}) is non-negative, and

(ii) the Riccati representatives $u_{\pm}$ obey $u_{\pm}\in L^{1}%
(\mathbb{R}^{\pm})\cap L^{2}(\mathbb{R})$.\newline We have $\mathcal{Q}%
=\mathcal{Q}_{0}\cup\mathcal{Q}_{>}$ where $\mathcal{Q}_{0}$ is the set of all
$q\in Q$ with $v(0)=0$, and $\mathcal{Q}_{>}$ is the set of all such
distributions with $v(0)>0$. This class includes the usual Faddeev--Marchenko
class\footnote{That is, real-valued measurable functions $q$ with $\int\left(
1+\left\vert x\right\vert \right)  \left\vert q(x)\right\vert ~dx<\infty$.}
but also positive measures with suitable decay, certain highly oscillating
potentials, and sums of delta functions with positive weight (see \S 1 of
\cite{FHMP:2009} and \S 2 of \cite{HMP:2010} for further examples). The set
$\mathcal{Q}_{0}$ is very unstable under perturbations so that potentials in
the sets $\mathcal{Q}_{0}$ and $\mathcal{Q}_{>}$ are referred to respectively
as \textquotedblleft exceptional\textquotedblright\ and \textquotedblleft
generic\textquotedblright\ potentials.

A distribution $q\in\mathcal{Q}$ is uniquely determined by the data%

\begin{equation}
\left(  \left.  u_{-}\right\vert _{(-\infty,0)},\left.  u_{+}\right\vert
_{(0,\infty)},v(0)\right)  \in X^{-}\times X^{+}\times\lbrack0,\infty
),\label{eq:vbl.riccati}%
\end{equation}
where $X^{\pm}=L^{2}(\mathbb{R}^{\pm})\cap L^{1}(\mathbb{R}^{\pm})$ (see
\cite{HMP:2010}, Lemma 2.3). We will call the triple $\left(  \left.
u_{-}\right\vert _{(-\infty,0)},\left.  u_{+}\right\vert _{(0,\infty
)},v(0)\right)  $ the \emph{Riccati representation} of $q$. Note that
$q\in\mathcal{Q}_{0}$ has a unique Riccati representative $u\in L^{1}%
(\mathbb{R})\cap L^{2}(\mathbb{R})$.

For $q\in\mathcal{Q}$, it was shown in \cite{FHMP:2009} (for $q\in
\mathcal{Q}_{0}$) and \cite{HMP:2010} (for $q\in\mathcal{Q}_{>}$) that the
usual formulation of scattering theory for the Schr\"{o}dinger equation
carries over. First, there exist Jost solutions $f_{\pm}(x,k)$, asymptotic as
$x\rightarrow\pm\infty$  to $\exp\left(  \pm ikx\right)  $. Second, one can
use these solutions to define reflection coefficients $r_{\pm}(k)$ that
describe scattering. The scattering maps $\mathcal{S}_{\pm}\,$\ are then
defined as%
\[
\mathcal{S}_{\pm}:q \mapsto r_{\pm}.
\]
We will study the scattering maps, parameterizing their domain using the
Riccati representation.

The Riccati representation connects the scattering problem for the
Schr\"{o}dinger equation to the scattering problem for the ZS--AKNS\ system
(see Zakharov--Shabat \cite{ZS:1971} and Ablowitz--Kaup--Newell--Segur
\cite{AKNS:1974}):%
\begin{equation}
\frac{d}{dx}\Psi=ik\sigma_{3}\Psi+Q(x)\Psi,\label{eq:AKNS}%
\end{equation}
where%
\[
\sigma_{3}=\left(
\begin{array}
[c]{cc}%
1 & 0\\
0 & -1
\end{array}
\right)
\]
and
\begin{equation}
Q(x)=\left(
\begin{array}
[c]{cc}%
0 & u(x)\\
u(x) & 0
\end{array}
\right)  ,\label{eq:Q}%
\end{equation}
where $u$ is a Riccati representative for $q$. If $q\in\mathcal{Q}_{0},$ then
the Schr\"{o}dinger scattering problem is in fact equivalent to the scattering
problem for (\ref{eq:AKNS}) with potential (\ref{eq:Q}), and the scattering
maps can be studied using techniques developed for the ZS--AKNS\ system (see
\cite{Frayer:2008} and \cite{FHMP:2009}). On the other hand, if $q\in
\mathcal{Q}_{>}$, one can construct Jost solutions $f_{+}$ and $f_{-}$ for the
Schr\"{o}dinger equation from scattering solutions associated to
ZS--AKNS\ systems (\ref{eq:AKNS}), where the potential $Q$ is given by
(\ref{eq:Q}) respectively with $u=u_{+}$ and $u=u_{-}$.

The Riccati representation is particularly well-suited to studying Sobolev
space mapping properties of the scattering map. We first consider the case of
$q\in\mathcal{Q}_{0}$, where $q$ is specified uniquely by a single real-valued
Riccati representative $u\in X$, with $X$ denoting the Banach space
$L^{1}(\mathbb{R})\cap L^{2}(\mathbb{R})$ with norm
\[
\left\Vert u\right\Vert _{X}=\left\Vert u\right\Vert _{L^{1}(\mathbb{R}%
)}+\left\Vert u\right\Vert _{L^{2}(\mathbb{R})}.
\]
We will write $X_{\mathbb{R}}$ for the real Banach space of real-valued
functions $u\in X$. Denote by $\widehat{X}$ and $\widehat{X_{\mathbb{R}}}$ the
images of $X$ and $X_{\mathbb{R}}$ under the Fourier transform, set
$\left\Vert \widehat{v}\right\Vert _{\widehat{X}}=\left\Vert v\right\Vert
_{X}$, and let%
\[
\widehat{X}_{1}:=\left\{  r\in\widehat{X_{\mathbb{R}}}:\left\Vert r\right\Vert
_{\infty}<1\right\}  .
\]
Note that $r(-k)=\overline{r(k)}$ for any $r\in\widehat{X_{\mathbb{R}}}$. It
was shown in \cite{Frayer:2008}, \cite{FHMP:2009} that the scattering maps
$\mathcal{S}_{\pm}$ are invertible, locally bi-Lipschitz maps from
$X_{\mathbb{R}}$ onto $\widehat{X}_{1}$. Since the maps $\mathcal{S}_{\pm}$ in
the Riccati variable are scattering maps for the ZS--AKNS\ system, one can use
techniques of Zhou \cite{Zhou:1998} to prove the following refined Sobolev
mapping property. For $s\geq0$, let
\[
L^{2,s}(\mathbb{R}):=\left\{  u\in L^{2}(\mathbb{R}):\left(  1+\left\vert
x\right\vert \right)  ^{s}u\in L^{2}(\mathbb{R})\right\}
\]
and denote by $H^{s}(\mathbb{R})$ the image of $L^{2,s}(\mathbb{R})$ under the
Fourier transform. Note that, for $s>1/2$, $L^{2,s}(\mathbb{R})\subset X$ and
$H^{s}(\mathbb{R})$ consists of continuous functions. If we set%
\[
H_{1}^{s}(\mathbb{R}):=\left\{  r\in H^{s}(\mathbb{R})\cap\widehat
{X_{\mathbb{R}}}:\left\Vert r\right\Vert _{\infty}<1\right\}  ,
\]
one can prove the following refined mapping property.

\begin{theorem}
\label{thm:except}For any $s>1/2$, the restrictions $\mathcal{S}_{\pm}%
:L^{2,s}(\mathbb{R})\cap X_{\mathbb{R}}\rightarrow H_{1}^{s}(\mathbb{R})$ are
onto, invertible, locally bi-Lipschitz continuous maps.
\end{theorem}

We will not give the details of the proof but rather concentrate on the more
challenging case where $q\in\mathcal{Q}_{>}$. To formulate our main theorem we
first recall some results from \cite{HMP:2010}.

If $q\in\mathcal{Q}_{>}$, the reflection coefficients $r_{\pm}$ belong to
$\widehat{X_{\mathbb{R}}}$ , but $r_{\pm}(0)=-1$ and $\left\vert
r(k)\right\vert <1$ for $k\neq0$. For smooth, compactly supported generic
potentials, one has $r_{\pm}(k)=-1+\mathcal{O}(k^{2})$ as $k\rightarrow0$
(see, for example, \cite{DT:1979}, \S 2, Theorem 1, Part V and Remark 9); in
general, as shown in \cite{HMP:2010}, one has the weaker condition that the
functions%
\[
\frac{1-\left\vert r_{\pm}(k)\right\vert ^{2}}{k^{2}}%
\]
belong to $\widehat{X_{\mathbb{R}}}$ and do not vanish at $k=0$. The direct
scattering maps in the Riccati variables are given by%
\[
\mathcal{S}_{\pm}:\left(  \left.  u_{-}\right\vert _{(-\infty,0)},\left.
u_{+}\right\vert _{(0,\infty)},v(0)\right)  \mapsto r_{\pm}.
\]
For $r\in\widehat{X_{\mathbb{R}}}$, we shall write%
\begin{equation}
\widetilde{r}(k)=\frac{1-\left\vert r(k)\right\vert ^{2}}{k^{2}}%
\label{eq:tilder-def}%
\end{equation}
and denote%
\[
\mathcal{R}_{>}:=\left\{  r\in\widehat{X_{\mathbb{R}}}:r(0)=-1,~\left\vert
r(k)\right\vert <1\text{ if }k\neq0\text{, }~\widetilde{r}\in\widehat
{X_{\mathbb{R}}},~\widetilde{r}(0)\neq0\right\}  .
\]
The space $\mathcal{R}_{>}$ is a metric space when equipped with the metric
\begin{equation}
d\left(  r_{1},r_{2}\right)  =\left\Vert r_{1}-r_{2}\right\Vert _{\widehat{X}%
}+\left\Vert \widetilde{r}_{1}-\widetilde{r}_{2}\right\Vert _{\widehat{X}%
}.\label{eq:dx}%
\end{equation}
In \cite{HMP:2010}, it was shown that the maps $\mathcal{S}_{\pm}$ are locally
bi-Lipschitz continuous onto maps from $X^{-}\times X^{+}\times(0,\infty)$
onto $\mathcal{R}_{>}$ equipped with the metric (\ref{eq:dx}).

We will prove a finer mapping property, analogous to Theorem \ref{thm:except},
for the scattering map on generic potentials. We set%
\[
\mathcal{R}_{s}=\left\{  r\in\mathcal{R}_{>}\cap H^{s}(\mathbb{R}%
):\widetilde{r}\in H^{s}(\mathbb{R})\right\}
\]
and equip $\mathcal{R}_{s}$ with the metric%
\[
d_{s}\left(  r_{1},r_{2}\right)  =\left\Vert r_{1}-r_{2}\right\Vert
_{H^{s}(\mathbb{R})}+\left\Vert \widetilde{r}_{1}-\widetilde{r}_{2}\right\Vert
_{H^{s}(\mathbb{R})}.
\]

\begin{theorem}
\label{thm:main}For any $s>1/2$, the direct scattering maps $\mathcal{S}_{\pm
}$ are invertible, locally bi-Lipschitz continuous maps from $L^{2,s}%
(\mathbb{R}^{-})\times L^{2,s}(\mathbb{R}^{+})\times(0,\infty)$ onto the space
$\mathcal{R}_{s}$.
\end{theorem}

Fourier-type mapping properties of the map $q\mapsto r$ have been studied by
many authors, including Cohen \cite{Cohen:1982}, Deift and Trubowitz
\cite{DT:1979}, and Faddeev \cite{Faddeev:1964}. These authors impose weighted
$L^{1}$ assumptions on $q$ and obtain regularity results for $r$ in terms of
$\infty$-norms of $r$ and its derivatives. Kappeler and Trubowitz
\cite{KT:1986}, \cite{KT:1988} studied Sobolev space mapping properties of the
scattering map, defined as follows. Let $s(k)=2ikr(k)/t(k)$, where $r$ is the
reflection coefficient and $t$ is the transmission coefficient, and introduce
the weighted Sobolev spaces
\begin{align*}
H_{n,\alpha} &  =\left\{  f\in L^{2}:x^{\beta}\partial_{x}^{j}f\in
L^{2}\text{, }0\leq j\leq n,~0\leq\beta\leq\alpha\right\}  ,\\
H_{n,\alpha}^{\#} &  =\left\{  f\in H_{n,\alpha}:x^{\beta}\partial_{x}%
^{n+1}f\in L^{2},~1\leq\beta\leq\alpha\right\}  .
\end{align*}
Kappeler and Trubowitz show that the map $q\mapsto s$ takes potentials $q\in
H_{N,N}$ without bound states to scattering functions $s$ belonging to
$H_{N-1,N}^{\#}$ for $N\geq3$. They extend their results to potentials with
finitely many bound states in \cite{KT:1988}. They also prove analyticity and
investigate the differential of the scattering map.

Our results are similar to those of Kappeler and Trubowitz in that we study
$L^{2}$-based Sobolev spaces, which leads to a more symmetrical formulation of
the mapping properties. In our case, we examine the scattering map in the
Riccati variables (\ref{eq:vbl.riccati}) and so treat potentials which are
more singular than the class treated by Kappeler and Trubowitz. In a
subsequent paper \cite{HMP:2011}, we will extend the methods developed here to
consider mapping properties between weighted fractional Sobolev spaces which
preserve the KdV flow.

This paper is organized as follows. In section 2, we first review the
connection between Jost solutions to the Schr\"{o}dinger and
ZS--AKNS\ equations. In section 3 we obtain estimates on the direct scattering
map using a Fourier representation for the Jost solutions derived in
\cite{FHMP:2009}. In section 4, we use the representation formulas of
\cite{HMP:2010}, derived from Gelfand--Levitan--Marchenko equation for the
ZS--AKNS\ system, to analyze the inverse scattering map. Finally, in section
5, we give the proof of the main theorem.

\textbf{Acknowledgements}. The research of RH and YM was partially supported
by Deutsche Forschungsgemeinschaft under project 436 UKR 113/84. RH\ was
supported in part by NSF grant DMS-0408419, and PP\ was supported in part by
NSF grants DMS-0408419 and DMS-0710477.

\section{Schr\"{o}dinger Scattering and the ZS--AKNS System}

\label{sec:pre}

In this section, we recall how the Jost solutions and reflection coefficients
for a Schr\"{o}dinger operator with Miura potential may be computed by solving
the associated ZS--AKNS\ equations with potentials $u_{+}$ and $u_{-}$. We
assume throughout that $u_{\pm}\in L^{2}(\mathbb{R})\cap L^{1}(\mathbb{R}%
^{\pm})$ are real-valued.

First, we recall the connection between the Schr\"{o}dinger equation with a
Miura potential and the ZS--AKNS\ system. If $u\in L_{\mathrm{loc}}%
^{2}(\mathbb{R})$ and $q=u^{\prime}+u^{2}$ then the Schr\"{o}dinger equation
\begin{equation}
-y^{\prime\prime}+qy=k^{2}y\label{eq:se}%
\end{equation}
is equivalent to the system%
\begin{equation}
\frac{d}{dx}\left(
\begin{array}
[c]{c}%
y\\
y^{\left[  1\right]  }%
\end{array}
\right)  =\left(
\begin{array}
[c]{cc}%
u & 1\\
-k^{2} & -u
\end{array}
\right)  \left(
\begin{array}
[c]{c}%
y\\
y^{\left[  1\right]  }%
\end{array}
\right)  \label{eq:se.sys}%
\end{equation}
where $y^{\left[  1\right]  }:=y^{\prime}-uy$ is the quasi-derivative of $y$.
Note that $y$ and $y^{\left[  1\right]  }$ are absolutely continuous, and the
initial value problem for (\ref{eq:se.sys}) has a unique solution. For a given
choice of $u$ and solutions $g$ and $h$ of (\ref{eq:se}), the Wronskian%
\begin{equation}
\left[  f,g\right]  =g(x)h^{\left[  1\right]  }(x)-g^{\left[  1\right]
}(x)h(x)\label{eq:Wronski}%
\end{equation}
is independent of $x$.

The Jost solutions $f_{\pm}(x,k)$ satisfy (\ref{eq:se}) with respective
asymptotic conditions%
\begin{equation}
\lim_{x\rightarrow\pm\infty}\left(
\begin{array}
[c]{c}%
f_{\pm}(x)-e^{\pm ikx}\\
f_{\pm}^{\left[  1\right]  }(x)\mp ike^{\pm ikx}%
\end{array}
\right)  =0\label{eq:fpm.asy}%
\end{equation}
where%
\[
f_{\pm}^{\left[  1\right]  }:=f_{\pm}-u_{\pm}f_{\pm}.
\]
If $\left[  ~\cdot~,~\cdot~\right]  _{\pm}$ denotes the Wronskian
(\ref{eq:Wronski}) with $u=u_{\pm}$, it follows from the asymptotics
(\ref{eq:fpm.asy}) that
\[
-\left[  f_{+}(x,k),f_{+}(x,-k)\right]  _{+}=\left[  f_{-}(x,k),f_{-}%
(x,-k)\right]  _{-}=2ik.
\]
Thus, for $k\neq0$, there are coefficients $a(k)$ and $b(k)$ so that%
\[
f_{+}(x,k)=a(k)f_{-}(x,-k)+b(k)f_{-}(x,k).
\]
By standard arguments,%
\begin{equation}
\left\vert a(k)\right\vert ^{2}-\left\vert b(k)\right\vert ^{2}%
=1,\label{eq:ab1}%
\end{equation}
and the reality conditions%
\begin{equation}
a(-k)=\overline{a(k)},%
\qquad
b(-k)=\overline{b(k)}\label{eq:abr}%
\end{equation}
hold. Moreover,%
\begin{equation}
a(k)=\frac{\left[  f_{+}(x,k),f_{-}(x,k)\right]  _{-}}{\left[  f_{-}%
(x,-k),f_{-}(x,k)\right]  _{-}}\label{eq:a.wronski}%
\end{equation}
and%
\begin{equation}
b(k)=\frac{\left[  f_{+}(x,k),f_{-}(x,-k)\right]  _{-}}{\left[  f_{-}%
(x,k),f_{-}(x,-k)\right]  _{-}}.\label{eq:b.wronski}%
\end{equation}
The reflection coefficients $r_{\pm}$ are given by%
\begin{align}
r_{-}(k) &  =b(k)/a(k),\label{eq:r-}\\
r_{+}(k) &  =-b(-k)/a(k),\label{eq:r+}%
\end{align}
so that $\left\vert r_{+}(k)\right\vert =\left\vert r_{-}(k)\right\vert $.
\ The transmission coefficient is given by $t(k)=1/a(k)$, and the involution%
\begin{equation}
r(k)\mapsto-\frac{t(k)}{t(-k)}r(-k)\label{eq:inv}%
\end{equation}
maps $r_{-}$ to $r_{+}$ and vice versa.

To compute the Jost solutions $f_{\pm}$ we exploit the following connection
between the Schr\"{o}dinger equation with potential $q=u^{\prime}+u^{2}$ and
the ZS--AKNS\ system%
\begin{equation}
\frac{d}{dx}\Psi=ik\sigma_{3}\Psi+Q(x)\Psi\label{eq:ZS-AKNS}%
\end{equation}
with potential%
\begin{equation}
Q(x)=\left(
\begin{array}
[c]{cc}%
0 & u(x)\\
u(x) & 0
\end{array}
\right)  .\label{eq:Qu}%
\end{equation}
If $\Psi=(\psi_{1},\psi_{2})^{T}$ is a vector-valued solution of
(\ref{eq:ZS-AKNS}) with potential (\ref{eq:Qu}), then
\[
\left(
\begin{array}
[c]{c}%
\psi_{1}+\psi_{2}\\
ik(\psi_{1}-\psi_{2})
\end{array}
\right)
\]
solves the system (\ref{eq:se.sys}). In particular, if $\Psi_{+}$ and
$\Psi_{-}$ are the unique matrix-valued solutions of the respective problems%
\begin{gather*}
\frac{d}{dx}\Psi^{\pm}=ik\sigma_{3}\Psi^{\pm}+\left(
\begin{array}
[c]{cc}%
0 & u_{\pm}(x)\\
u_{\pm}(x) & 0
\end{array}
\right)  \Psi^{\pm},\\
\lim_{x\rightarrow\pm\infty}\left\vert \Psi^{\pm}(x)-e^{ixk\sigma_{3}%
}\right\vert =0,
\end{gather*}
then the formulas%
\begin{align}
f_{+}(x,k) &  =\psi_{11}^{+}(x,k)+\psi_{21}^{+}(x,k)\label{eq:f+}\\
f_{+}^{\left[  1\right]  }(x,k) &  =ik\left(  \psi_{11}^{+}(x,k)-\psi_{21}%
^{+}(x,k)\right)  \label{eq:f+[1]}\\
f_{-}(x,k) &  =\overline{\psi_{11}^{-}(x,k)}+\overline{\psi_{21}^{-}%
(x,k)}\label{eq:f-}\\
f_{-}^{\left[  1\right]  }(x,k) &  =-ik\left(  \overline{\psi_{11}^{-}%
(x,k)}-\overline{\psi_{21}^{-}(x,k)}\right)  \label{eq:f-[1]}%
\end{align}
hold, where the bar denotes complex conjugation. A short computation with
(\ref{eq:a.wronski})--(\ref{eq:b.wronski}) leads to the formulas%
\begin{align}
a(k) &  =\left\vert
\begin{array}
[c]{cc}%
\psi_{11}^{+}(x,k) & \overline{\psi_{21}^{-}(x,k)}\\
& \\
\psi_{21}^{+}(x,k) & \overline{\psi_{11}^{-}(x,k)}%
\end{array}
\right\vert -\frac{v(x)}{2ik}f_{+}(x,k)f_{-}(x,k),\label{eq:a}\\
& \nonumber\\
b(k) &  =\left\vert
\begin{array}
[c]{cc}%
\psi_{11}^{+}(x,k) & \overline{\psi_{11}^{+}(x,-k)}\\
& \\
-\psi_{21}^{+}(x,k) & -\overline{\psi_{21}^{-}(x,-k)}%
\end{array}
\right\vert +\frac{v(x)}{2ik}f_{+}(x,k)f_{-}(x,k),\label{eq:b}%
\end{align}
where, for a $2\times2$ matrix $A$, $\left\vert A\right\vert $ denotes the determinant.

These two formulas lie at the heart of our analysis for the direct problem.
They show explicitly the singularity at $k=0$ that occurs when $u_{+}\neq
u_{-}$; the singularity is always nonzero in this case since $v$ is strictly
nonzero and $f_{\pm}(x,0)$ are positive solutions of the zero-energy
Schr\"{o}dinger equation.

To study the scattering map via the formulas (\ref{eq:a})--(\ref{eq:b}), we
will use integral representations for the solutions $\Psi^{\pm}$. \ These
integral representations give $\Psi^{\pm}$ as Fourier transforms of functions
given by explicit multilinear series in $u_{\pm}$. Let~$\Psi^{\pm}%
(x,k)=\exp(ixk\sigma_{3})N^{\pm}(x,k)$ and denote by~$n_{ij}^{\pm}$ the
entries of~$N^{\pm}$. In order to compute the Jost solutions from
(\ref{eq:f+})--(\ref{eq:f-[1]}), it suffices to study~$n_{11}^{\pm}$ and
$n_{21}^{\pm}$. We will describe only the integral representations for
$n_{11}^{+}$ and $n_{21}^{+}$ and their properties since those of $n_{11}^{-}$
and $n_{21}^{-}$ are very similar.

From \cite{FHMP:2009}, section 3.1, equations (3.14)\ and (3.15) and
following, we have%
\begin{align*}
n_{11}^{+}(x,k)-1  &  =\int_{0}^{\infty}A(x,\zeta)e^{i\zeta k}~d\zeta,\\
n_{21}^{+}(x,k)  &  =\int_{x}^{\infty}B(x,\zeta)e^{i\zeta k}~d\zeta.
\end{align*}
Here $A$ and $B$ have multilinear expansions of the form%
\[
A(x,\zeta)=\sum_{n=1}^{\infty}A_{n}(x,\zeta), \qquad B(x,\zeta)=\sum
_{n=1}^{\infty}B_{n}(x,\zeta)
\]
with%
\[
A_{n}(x,\zeta)=\int_{\Omega_{2n}(\zeta)}u_{+}(y_{1})\ldots u_{+}%
(y_{2n})~dS_{2n}
\]
and%
\[
B_{n}(x,\zeta)=\int_{\Omega_{2n-1}(\zeta)}u_{+}(y_{1})\ldots u_{+}%
(y_{2n})~dS_{2n-1},
\]
where, for $\zeta\in\mathbb{R}$, $\Omega_{n}(\zeta)$ is the set of all
$\mathrm{y}=\left(  y_{1},\ldots,y_{n}\right)  $ in $\mathbb{R}^{n}$ with
$x\leq y_{1}\leq\ldots\leq y_{n}$ and
\begin{equation}
\sum_{j=0}^{n-1}(-1)^{j}y_{n-j}=\zeta, \label{eq:zeta}%
\end{equation}
while $dS_{n}$ is surface measure on the hyperplane (\ref{eq:zeta}).

For each fixed $x$ we have
\begin{align}
\left\Vert n_{11}^{+}(x,~\cdot~)-1\right\Vert _{H^{s}(\mathbb{R})} &
\leq\left\Vert A(x,~\cdot~)\right\Vert _{L^{2,s}(\mathbb{R})}%
,\label{eq:n11+.to.A}\\
\left\Vert n_{21}^{+}(x,~\cdot~)\right\Vert _{H^{s}(\mathbb{R})} &
\leq\left\Vert B(x,~\cdot~)\right\Vert _{L^{2,s}(\mathbb{R})}%
.\label{eq:n21+.to.B}%
\end{align}
Thus, to estimate the $H^{s}$-norms of $n_{11}^{+}$ and $n_{21}^{+}$ as
functions of $k$, it suffices to obtain summable estimates on $\left\Vert
A_{n}(x,~\cdot~)\right\Vert _{L^{2,s}(\mathbb{R})}$ and $\left\Vert
B_{n}(x,~\cdot~)\right\Vert _{L^{2,s}(\mathbb{R})}$.

To do this, we first note the identity
\[
\left\Vert \psi\right\Vert _{L^{2,s}(\mathbb{R})}=\sup\left\{  \left\vert
\int\varphi(\zeta)(1+\left\vert \zeta\right\vert )^{s}\psi(\zeta
)~d\zeta\right\vert :\left\Vert \varphi\right\Vert _{L^{2}}=1\right\}  .
\]
Next, setting ${\mathrm{y}}:=(y_{1},\dots,y_{n})$, ${\mathrm{d}}{\mathrm{y}%
}:=dy_{1}\cdots dy_{n}$, $U({\mathrm{y}}):=u_{+}(y_{1})\cdots u_{+}(y_{n})$,
and defining
\[
\zeta_{n}({\mathrm{y}}):=\sum_{j=0}^{n-1}(-1)^{j}y_{n-j},
\]
we find that
\begin{gather}
\int_{0}^{\infty}\varphi(\zeta)\int_{\Omega_{2n}(\zeta)}U({\mathrm{y}}%
)dS_{2n}d\zeta=\int_{x\leq y_{1}\leq\dots\leq y_{2n}}U({\mathrm{y}}%
)\varphi(\zeta_{2n}({\mathrm{y}})){\mathrm{d}}{\mathrm{y}},\label{eq:8.1}\\
\int_{x}^{\infty}\varphi(\zeta)\int_{\Omega_{2n-1}(\zeta)}U({\mathrm{y}%
})dS_{2n-1}d\zeta=\int_{x\leq y_{1}\leq\dots\leq y_{2n-1}}U({\mathrm{y}%
})\varphi(\zeta_{2n-1}({\mathrm{y}})){\mathrm{d}}{\mathrm{y}}.\label{eq:8.2}%
\end{gather}
Observe that for ${\mathrm{y}}$ obeying $0\leq x\leq y_{1}\leq\cdots\leq
y_{n}$ the estimate $|\zeta_{n}({\mathrm{y}})|\leq y_{n}$ holds. We then get
from the integral representation for $A_{n}$ and \eqref{eq:8.1} that, for any
$\varphi\in L^{2}({\mathbb{R}})$,
\begin{align*}
&  \Bigl|\int_{0}^{\infty}(1+\zeta)^{s}\varphi(\zeta)A_{n}(x,\zeta
)\,d\zeta\Bigr|\\
&  \leq\int_{x\leq y_{1}\leq\dots\leq y_{2n-1}}|u_{+}(y_{1})\cdots
u_{+}(y_{2n-1})|\int_{x}^{\infty}(1+y_{2n})^{s}|\varphi(\zeta_{2n}%
({\mathrm{y}}))||u_{+}(y_{2n})|\,{\mathrm{d}}{\mathrm{y}}\\
&  \leq\frac{\Vert u_{+}\Vert_{L^{1}(\mathbb{R}^{+})}^{2n-1}}{(2n-1)!}\Vert
u_{+}\Vert_{L^{2,s}({\mathbb{R}}^{+})}\Vert\varphi\Vert_{L^{2}({\mathbb{R}})}.
\end{align*}
Therefore
\[
\Vert A_{n}(x,\,\cdot\,)\Vert_{L^{2,s}({\mathbb{R}}^{+})}\leq\frac{\Vert
u_{+}\Vert_{L^{1}(\mathbb{R}^{+})}^{2n-1}}{(2n-1)!}\Vert u_{+}\Vert
_{L^{2,s}({\mathbb{R}}^{+})},
\]
and similar estimates give
\[
\Vert B_{n}(x,\,\cdot\,)\Vert_{L^{2,s}({\mathbb{R}}^{+})}\leq\frac{\Vert
u_{+}\Vert_{L^{1}(\mathbb{R}^{+})}^{2n-2}}{(2n-2)!}\Vert u_{+}\Vert
_{L^{2,s}({\mathbb{R}}^{+})}.
\]
Since $A_{n}(x,\,\cdot\,)$ and $B_{n}(x,\,\cdot\,)$ are multilinear functions
of~$u_{+}$ and the series for $A(x,\,\cdot\,)$ and $B(x,\,\cdot\,)$ converge
absolutely in~$L^{2,s}({\mathbb{R}})$, standard arguments show that, for every
fixed $x\geq0$, $A(x,\,\cdot\,)$ and $B(x,\,\cdot\,)$ depend analytically in
$L^{2,s}({\mathbb{R}})$ on~$u_{+}\in X^{+}$. Hence:

\begin{proposition}
\label{prop:n+.Hs}Assume that $s>1/2$ and that $u_{+}\in L^{2,s}({\mathbb{R}%
}^{+})$. Then $n_{11}^{+}(x,\,\cdot\,)-1$ and $n_{21}^{+}(x,\,\cdot\,)$ belong
to $H^{s}(\mathbb{R})$ for each fixed~$x\geq0$, depend analytically therein
on~$u_{+} \in X^{+}$, and the estimates
\[
\sup_{x\geq0} \bigl(\Vert n_{11}^{+}(x,\,\cdot\,)-1 \Vert_{H^{s}(\mathbb{R})}
+ \Vert n_{21}^{+}(x,\,\cdot\,)\Vert_{H^{s}(\mathbb{R})}\bigr)
\leq\Vert u_{+}\Vert_{L^{2,s}({\mathbb{R}}^{+})} \exp\{\Vert u_{+}\Vert
_{L^{1}({\mathbb{R}}^{+})}\}
\]
hold.
\end{proposition}

A similar analysis, based on the integral representations for $n_{11}^{-}$ and
$n_{21}^{-}$, shows:

\begin{proposition}
\label{prop:n-.Hs}Assume that $s>1/2$ and that $u_{-}\in L^{2,s}({\mathbb{R}%
}^{-})$. Then $n_{11}^{-}(x,\,\cdot\,)-1$ and $n_{21}^{-}(x,\,\cdot\,)$ belong
to $H^{s}(\mathbb{R})$ for each fixed~$x\leq0$, depend analytically therein
on~$u_{-} \in X^{-}$, and the estimates
\[
\sup_{x\leq0}\bigl(\Vert n_{11}^{-}(x,\,\cdot\,)-1\Vert_{H^{s}(\mathbb{R})}
+\Vert n_{21}^{-}(x,\,\cdot\,)\Vert_{H^{s}(\mathbb{R})}\bigr)
\leq\Vert u_{-}\Vert_{L^{2,s}({\mathbb{R}}^{-})} \exp\{\Vert u_{-}\Vert
_{L^{1}({\mathbb{R}}^{-})}\}
\]
hold.
\end{proposition}

\section{The Direct Problem}

\label{sec:direct}

We now consider the mappings $\left(  u_{-},u_{+},v(0)\right)  \mapsto r_{\pm
}$. In order to study the mapping properties we introduce the auxiliary
functions%
\begin{align}
\widetilde{a}(k)  &  =\frac{k}{k+i}a(k),\label{eq:tildea}\\
\widetilde{b}(k)  &  =\frac{k}{k+i}b(k),\label{eq:tildeb}\\
\widetilde{r}(k)  &  =\frac{1-\left\vert r_{\pm}(k)\right\vert ^{2}}{k^{2}},
\label{eq:tilder}%
\end{align}
and note the relations%
\begin{equation}
r_{-}(k)=\frac{\widetilde{b}(k)}{\widetilde{a}(k)}, \qquad r_{+}(k)=\frac
{i-k}{i+k}\frac{\widetilde{b}(-k)}{\widetilde{a}(k)} \label{eq:rpm}%
\end{equation}
and%
\begin{equation}
\widetilde{r}(k)=\frac{1}{k^{2}+1}\frac{1}{\left\vert \widetilde
{a}(k)\right\vert ^{2}}. \label{eq:tilder.rep}%
\end{equation}

\begin{proposition}
\label{prop:direct}Suppose that $u_{\pm}\in L^{2,s}(\mathbb{R}^{\pm})$ for
some $s>1/2$. Then $r_{\pm}\in H^{s}(\mathbb{R})$ and $\widetilde{r}\in
H^{s}(\mathbb{R})$ with $\widetilde{r}(0)\neq0$, and the maps
\begin{align*}
L^{2,s}(\mathbb{R}^{+})\times L^{2,s}(\mathbb{R}^{-})\times\mathbb{R}^{+} &
\rightarrow H^{s}(\mathbb{R})^{3}\\
\left(  u_{+},u_{-},v(0)\right)   &  \mapsto(r_{-},r_{+},\widetilde{r})
\end{align*}
are locally Lipschitz continuous.
\end{proposition}

\begin{proof}
From the representation formulae (\ref{eq:a}) and (\ref{eq:b}) evaluated at
$x=0$ we have%
\[
\widetilde{a}(k)=\frac{k}{k+i}\left\vert
\begin{array}
[c]{cc}%
n_{11}^{+}(0,k) & \overline{n_{21}^{-}(0,k)}\\
& \\
n_{21}^{+}(0,k) & \overline{n_{11}^{-}(0,k)}%
\end{array}
\right\vert -\frac{1}{k+i}\frac{v(0)}{2i}f_{+}(0,k)f_{-}(0,k)
\]
and
\[
\widetilde{b}(k)=\frac{k}{k+i}\left\vert
\begin{array}
[c]{cc}%
n_{11}^{+}(0,k) & \overline{n_{11}^{+}(0,-k)}\\
& \\
-n_{21}^{+}(0,k) & -\overline{n_{21}^{-}(0,-k)}%
\end{array}
\right\vert +\frac{1}{k+i}\frac{v(0)}{2i}f_{+}(0,k)f_{-}(0,k).
\]
In view of Propositions~\ref{prop:n+.Hs} and \ref{prop:n-.Hs} the functions
$n_{ij}^{\pm}(0,\,\cdot\,)$ and $f_{\pm}(0,\,\cdot\,)$ belong to the Banach
algebra $1\dotplus H^{s}({\mathbb{R}})$ (see Appendix~\ref{sec:app}) and
depend locally Lipschitz continuously therein on the Riccati variables
$(u_{+},u_{-},v(0))$; thus the same is true of $\widetilde{a}$ and
$\widetilde{b}$. Moreover, the function $\widetilde{a}$ is an invertible
element of $1\dotplus H^{s}({\mathbb{R}})$. Indeed, by
Lemma~\ref{lem:spectrum} it suffices to show that $\inf|\widetilde{a}|>0$. We
observe that $\widetilde{a}(0)=v(0)f_{+}(0,k)f_{-}(0,k)\neq0$ while
$|\widetilde{a}(k)|>0$ for all nonzero real~$k$ due to~\eqref{eq:ab1}.
Representation~\eqref{eq:a.wronski} for $a$ along with the asymptotic behavior
of the Jost solutions imply that $|a(k)|\rightarrow1$ as $k\rightarrow
\pm\infty$, so that $|\widetilde{a}(k)|\rightarrow1$ as $k\rightarrow\pm
\infty$ as well. Recalling that $\widetilde{a}$ is a continuous function, we
conclude that $\widetilde{a}$ is an invertible element of the Banach
algebra~$1\dotplus H^{s}({\mathbb{R}})$. Clearly, the same conclusion holds
for all $\widetilde{a}$ in a neighborhood of the given one.

We now use~\eqref{eq:rpm} to conclude that the reflection coefficients
$r^{\pm}$ belong to $H^{s}({\mathbb{R}})$ and depend therein locally Lipschitz
continuously on the Riccati variables. Since $|\widetilde{a}(k)|^{2}$ is an
invertible element of $1\dotplus H^{s}({\mathbb{R}})$,
relation~\eqref{eq:tilder.rep} yields the inclusion $\widetilde{r}\in
H^{s}({\mathbb{R}})$, and the continuous dependence follows by the same
arguments as above. Finally, \eqref{eq:tilder.rep} and $\widetilde{a}(0)\neq0$
yield $\widetilde{r}(0)\neq0$, and the proof is complete.
\end{proof}

Finally, we note the following variant of Proposition~3.3 of~\cite{HMP:2010},
which concerns continuity of the involution (\ref{eq:inv}) between reflection
coefficients. For a given $r\in\mathcal{R}_{s}$ with $\widetilde{r}$
of~\eqref{eq:tilder-def}, we define%
\[
t(z)=\frac{z}{z+i}\exp\left(  \frac{1}{2\pi i}\int_{\mathbb{R}}\log\left[
\left(  s^{2}+1\right)  \widetilde{r}(s)\right]  \frac{ds}{s-z}\right)
\]
for $\operatorname{Im}(z)>0$, and by the boundary value for real $z=k$.

\begin{proposition}
The mapping
\[
\mathcal{I}_{s}:r\mapsto-\frac{t(k)}{t(-k)}r(-k)
\]
is a continuous involution from $\mathcal{R}_{s}$ to itself.
\end{proposition}

We omit the proof, since it is completely analogous to that of Proposition~3.3
in~\cite{HMP:2010}, except that the Banach algebra $1\dotplus\widehat{X}$
there is replaced with the Banach algebra~$1\dotplus H^{s}({\mathbb{R}})$.

\section{The Inverse Problem}

\label{sec:inverse}

In this section, we assume given a function $r\in\mathcal{R}_{s}$ and set
$r^{\#}=\mathcal{I}_{s}r$. It follows from \cite{HMP:2010} that there exists a
unique distribution $q\in H^{-1}(\mathbb{R})$ with Riccati representatives
$u=u_{+}\in L^{2}(\mathbb{R})\cap L^{1}(\mathbb{R}^{+})$ and $u^{\#}=u_{-}\in
L^{2}(\mathbb{R})\cap L^{1}(\mathbb{R}^{-})$ so that the corresponding
Schr\"{o}dinger operator has $r$ and $r^{\#}$ as its right and left reflection
coefficients, respectively. We wish to show that the Riccati representatives
$u$ and $u^{\#}$ reconstructed from $r$ and $r^{\#}$ belong respectively to
$L^{2,s}(\mathbb{R}^{+})$ and $L^{2,s}(\mathbb{R}^{-})$. To do so, we will
recall the reconstruction formulas for $u$ and $u^{\#}$ derived in
\cite{HMP:2010} from the Gelfand--Levitan--Marchenko equations. Let us define%
\begin{align*}
F(x) &  =\frac{1}{\pi}\int_{-\infty}^{\infty}r(k)e^{2ikx}~dk,\\
F^{\#}(x) &  =\frac{1}{\pi}\int_{-\infty}^{\infty}r^{\#}(k)e^{-2ikx}dk.
\end{align*}
Note that $F$ and $F^{\#}$ belong to $L^{2,s}(\mathbb{R})$. Setting%
\[
\Omega(x)=\left(
\begin{array}
[c]{cc}%
0 & F(x)\\
F(x) & 0
\end{array}
\right)  ,~~\Omega^{\#}(x)=\left(
\begin{array}
[c]{cc}%
0 & F^{\#}(x)\\
F^{\#}(x) & 0
\end{array}
\right)  ,
\]
the right and left Gelfand--Levitan--Marchenko equations are respectively
\begin{align*}
\Omega(x+\zeta)+\Gamma(x,\zeta)+\int_{0}^{\infty}\Gamma(x,t)\Omega
(x+\zeta+t)~dt &  =0,%
\qquad
\zeta>0,\\
\Omega^{\#}(x+\zeta)+\Gamma^{\#}(x,\zeta)+\int_{0}^{\infty}\Gamma
^{\#}(x,t)\Omega^{\#}(x+\zeta+t)~dt &  =0,%
\qquad
\zeta<0,
\end{align*}
for the $2\times2$ matrix-valued kernels $\Gamma$ and $\Gamma^{\#}$. The right
and left Riccati representatives are reconstructed via%
\begin{align*}
u(x)  & =-\Gamma_{12}(x,0),\\
u^{\#}(x)  & =\Gamma_{12}^{\#}(x,0).
\end{align*}
Let $\gamma(x,\zeta)=\Gamma_{12}(x,\zeta)$ and $\gamma^{\#}(x,\zeta
)=\Gamma_{12}^{\#}(x,\zeta)$. Let $T_{F}$ and $T_{F^{\#}}$ be the integral
operators (depending parametrically on $x$)%
\begin{align*}
\left(  T_{F}\psi\right)  (\zeta) &  =\int_{0}^{\infty}F(x+\zeta
+t)\psi(t)~dt,\\
\left(  T_{F^{\#}}\psi\right)  (\zeta) &  =\int_{-\infty}^{0}F^{\#}%
(x+\zeta+t)\psi(t)~dt.
\end{align*}
Then, as vectors in $L^{2}(\mathbb{R}^{+})$ (resp. in $L^{2}(\mathbb{R}^{-})$)
for each fixed $x$,
\begin{align*}
\left(  I-T_{F}^{2}\right)  \gamma(x,~\cdot~) &  =-F(x+~\cdot~),\\
\left(  I-T_{F^{\#}}^{2}\right)  \gamma^{\#}(x,~\cdot~) &  =-F^{\#}%
(x+~\cdot~).
\end{align*}
As shown in the proof of \cite{HMP:2010}, Proposition 4.2, the operator
$\left(  I-T_{F}^{2}\right)  ^{-1}$ is bounded from $L^{2}(\mathbb{R}^{+})$ to
itself (resp. $\left(  I-T_{F^{\#}}^{2}\right)  ^{-1}$ is bounded from
$L^{2}(\mathbb{R}^{-})$ to itself). From these equations and the
reconstruction formulas, it is not difficult to see that
\begin{align*}
u(x) &  =F(x)-G(x),\\
u^{\#}(x) &  =-F^{\#}(x)+G^{\#}(x),
\end{align*}
where%
\begin{align*}
G(x) &  =\int_{0}^{\infty}F(x+t)H(x,t)~dt,\\
G^{\#}(x) &  =\int_{-\infty}^{0}F^{\#}(x+t)H^{\#}(x,t)~dt
\end{align*}
and
\begin{align*}
H(x,~\cdot~) &  =\left(  I-T_{F}^{2}\right)  ^{-1}\left(  (T_{F}%
F)(x+~\cdot~)\right)  ,\\
H^{\#}(x,~\cdot~) &  =\left(  I-T_{F^{\#}}^{2}\right)  ^{-1}\left(
(T_{F^{\#}}F^{\#})(x+~\cdot~)\right)  .
\end{align*}

We are interested in estimating the behavior of $G$ as $x\rightarrow+\infty$
(resp. of $G^{\#}$ as $x\rightarrow-\infty$). It suffices to consider
$x>x_{0}$ (resp. $x<-x_{0}$) for sufficiently large $x_{0}$. Choosing $x_{0}$
so large that
\[
\int_{x_{0}}^{\infty}\left\vert F(s)\right\vert ~ds<1/2,~\int_{-\infty}%
^{x_{0}}\left\vert F^{\#}(s)\right\vert ~ds<1/2,
\]
we have $\left\Vert T_{F}\right\Vert _{L^{p}\rightarrow L^{p}}<1$ for $p=1,2$,
and similarly for $T_{F^{\#}}$. Note that we can make such a choice of fixed
$x_{0}$ in a small neighborhood of a given $F\in L^{2,s}(\mathbb{R})$ since
$L^{2,s}(\mathbb{R})\subset L^{1}(\mathbb{R})$ for $s>1/2$. We can then obtain
convergent multilinear expansions for $G$ and $G^{\#}$ valid respectively for
$x>x_{0}$ and $x<-x_{0}$. These multilinear expansions can be estimated, much
as in the previous section, to obtain the required weighted estimates. We will
give the analysis for $G$ since the analysis for $G^{\#}$ is very similar.

For $x>x_{0}$ we have the expansion%
\[
H(x,~\cdot~)=\sum_{j=0}^{\infty}\left(  T_{F}^{2j+1}\left[  F(x+~\cdot
~)\right]  \right)  (~\cdot~)
\]
convergent in $L^{2}(\mathbb{R}^{+})$. From this expansion and the Cauchy--Schwarz
inequality it follows that
\[
G(x)=\sum_{n=1}^{\infty}G_{n}(x)
\]
in $L^{\infty}(x_{0},\infty)$, where%
\[
G_{n}(x)=\int_{\mathbb{R}_{+}^{2n}}F(x+t_{1})F(x+t_{1}+t_{2})\ldots
F(x+t_{2n-1}+t_{2n})F(x+t_{2n})~\mathrm{dt}%
\]
and $\mathrm{dt}:=dt_{1}\ldots dt_{2n}$. We will show that, for $x_{0}>0$,
\begin{equation}
    \int_{x_{0}}^{\infty}\left(  1+x\right)^{2s}
        \left\vert G_{n}(x)\right\vert^{2}~dx
    \leq\left\Vert F\right\Vert _{L^{1}(x_{0},\infty)}^{4n}
        \int_{x_{0}}^{\infty}\left(  1+x\right)^{2s}
        \left\vert F(x)\right\vert ^{2}~dx,
\label{eq:Gn.bound}%
\end{equation}
from which it follows that $\int_{0}^{\infty}\left(1+x\right)^{2s}
  \left\vert G(x)\right\vert^{2}~dx<\infty$. \ Let $f(x):=\left\vert
F(x)\right\vert $ and $\widetilde{f}(x):=\left(  1+x\right)  ^{2s}f^{2}(x)$.
Since $x\leq x+t_{1}$ in the range of integration for $G_{n}$, it follows from
the Cauchy--Schwarz inequality that%
\[
\int_{x_{0}}^{\infty}\left(  1+x\right)  ^{2s}\left\vert G_{n}(x)\right\vert
^{2}~dx\leq\int_{x_{0}}^{\infty}I_{n}(x)J_{n}(x)~dx,
\]
where%
\begin{equation}
I_{n}(x):=\int_{\mathbb{R}_{+}^{2n}}\widetilde{f}(x+t_{1})f(x+t_{1}%
+t_{2})\ldots f(x+t_{2n-1}+t_{2n})f(x+t_{2n})~\mathrm{dt}\label{eq:In}%
\end{equation}
and%
\begin{equation}
J_{n}(x):=\int_{\mathbb{R}_{+}^{2n}}f(x+t_{1}+t_{2})\ldots f(x+t_{2n-1}%
+t_{2n})f(x+t_{2n})~\mathrm{dt}.\label{eq:Jn}%
\end{equation}
Clearly,
\begin{equation}
J_{n}(x)\leq\left\Vert f\right\Vert _{L^{1}(x_{0},\infty)}^{2n}%
\label{eq:Jn.bound}%
\end{equation}
for $x\geq x_{0}$. In (\ref{eq:In}), set $y_{2k-1}=x+t_{2k-1}$ and
$y_{2k}=t_{2k}$ for $1\leq k\leq n$; then
\begin{multline*}
\int_{x_{0}}^{\infty}I_{n}(x)~dx=\\
\int_{x_{0}}^{\infty}\int_{0}^{\infty}\ldots\int_{0}^{\infty}\int_{x_{0}%
}^{\infty}\widetilde{f}(y_{1})f(y_{1}+y_{2})\ldots f(y_{2n-1}+y_{2n}%
)f(x+y_{2n})~\mathrm{dy}~dx
\end{multline*}
where $\mathrm{dy}:=dy_{1}\ldots dy_{2n}$. It follows easily that
\begin{equation}
\int_{x_{0}}^{\infty}I_{n}(x)~dx\leq\left(  \int_{x_{0}}^{\infty}\widetilde
{f}(x)~dx\right)  \left(  \int_{x_{0}}^{\infty}f(x)~dx\right)  ^{2n}%
.\label{eq:In.bound}%
\end{equation}
Combining (\ref{eq:In.bound}) and (\ref{eq:Jn.bound}) gives (\ref{eq:Gn.bound}).

Together with a similar analysis for $u^{\#}$ and $G^{\#}$, the above
arguments yield:

\begin{proposition}
\label{prop:inverse}Suppose that $r\in\mathcal{R}_{s}$ for $s>1/2$. Then $u\in
L^{2,s}(\mathbb{R}^{+})$ and $u^{\#}\in L^{2,s}(\mathbb{R}^{-})$, and the maps
$r\mapsto u$ and $r^{\#}\mapsto u^{\#}$ are locally Lipschitz continuous
respectively as maps $\mathcal{R}_{s}\rightarrow L^{2,s}(\mathbb{R}^{+})$ and
$\mathcal{R}_{s}\rightarrow L^{2,s}(\mathbb{R}^{-})$.
\end{proposition}

\section{Proof of the Main Theorem}

\label{sec:proof}

We now give the proof of Theorem \ref{thm:main}. Proposition \ref{prop:direct}
shows that $\mathcal{S}_{\pm}$ have range contained in $\mathcal{R}_{s}$ and
that $\mathcal{S}_{\pm}$ are locally Lipschitz continuous maps from
$L^{2,s}(\mathbb{R}^{-})\times L^{2,s}(\mathbb{R}^{+})\times(0,\infty)$ into
the space $\mathcal{R}_{s}$. On the other hand, given a reflection coefficient
$r\in\mathcal{R}_{s}$, Proposition \ref{prop:inverse} shows that the Riccati
representatives reconstructed from $r$ and $r^{\#}$ satisfy $u\in
L^{2,s}(\mathbb{R}^{+})$ and $u^{\#}\in L^{2,s}(\mathbb{R}^{-})$ and are
locally Lipschitz continuous as respective functions of $r$ and $r^{\#}$. It
follows from the analysis of section 4 in \cite{HMP:2010} that $u$ and
$u^{\#}$ are the unique right- and left-hand Riccati representatives of a
real-valued distribution $q\in H^{-1}(\mathbb{R})$ having reflection
coefficients $r$ and $r^{\#}$. This shows that $\mathcal{S}_{\pm}$ are onto
$\mathcal{R}_{s}$ and completes the proof of Theorem \ref{thm:main}.

\appendix

\section{$H^{s}({\mathbb{R}})$ as a Banach algebra}

\label{sec:app}

Throughout this appendix, we shall write $L^{p}$ and $H^{s}$ for the
spaces~$L^{p}({\mathbb{R}})$ and $H^{s}({\mathbb{R}})$, respectively. We refer
the reader to the book by Runst and Sickel \cite{RS:1996} for the properties
of the Sobolev spaces $H^{s}$ and to the book by Rudin \cite{Rudin:1973} for
the basic notions of the Banach algebras.

For $s>\tfrac{1}{2}$, the space $H^{s}$ is a closed algebra with respect to
pointwise addition and multiplication. Thus, upon introducing an equivalent
norm, $H^{s}$ becomes a Banach algebra. We denote by $1\dotplus H^{s}$ the
extension of $H^{s}$ to a unital algebra; $1\dotplus H^{s}$ consists of
functions of the form $g:=c\cdot1+f$ with $c\in\mathbb{C}$ and $f\in H^{s}$.
We recall that the \emph{spectrum} $\sigma(g)$ of an element $g\in1\dotplus
H^{s}$ is the set of all $\lambda\in\mathbb{C}$ such that $g-\lambda\cdot1$ is
not invertible in~$1\dotplus H^{s}$.

\begin{lemma}
\label{lem:spectrum} Assume that $s>\tfrac12$. Then for every $g \in1 \dotplus
H^{s}$, the spectrum $\sigma(g)$ is contained in the closure $\overline
{\operatorname{ran}g}$ of the range $\operatorname{ran}g$.
\end{lemma}

\begin{proof}
It suffices to prove the implication
\[
g\in1\dotplus H^{s}\quad\text{and}\quad0\not \in \overline{\operatorname{ran}%
g}\implies\frac{1}{g}\in1\dotplus H^{s}.
\]
Without loss of generality, we may assume that $g=1+f$ with $f\in H^{s}$.
Also, we set $C:=\Vert1/g\Vert_{L^{\infty}}$.

Consider first the case~$s \in(\tfrac12,1)$. Recall that then $\phi\in L^{2}$
belongs to $H^{s}$ if and only if
\[
\int_{\mathbb{R}} \int_{\mathbb{R}} \frac{|\phi(x) - \phi(y)|^{2}%
}{|x-y|^{1+2s}}\,dx\,dy < \infty.
\]
Setting
\[
\phi(x) = \frac1{g(x)}-1 = - \frac{f(x)}{g(x)} \in L^{2}
\]
and observing that
\[
|\phi(x) - \phi(y)| \le C^{2} |f(x)-f(y)|,
\]
we easily conclude that $\phi\in H^{s}$.

Next, for~$s=1$ we find that
\[
\Bigl(\frac1g\Bigr)^{\prime}= - \frac{f^{\prime}}{g^{2}} \in L^{2},
\qquad\frac1g -1 = - \frac{f}{g} \in L^{2},
\]
so that $1/g \in1 \dotplus H^{s}$.

Finally, let $s=n+\alpha$, where $n\in\mathbb{N}$ and $\alpha\in(0,1)$. Then
\[
\Bigl(\frac{1}{g}\Bigr)^{(n)}=-\frac{f^{(n)}}{g^{2}}+\psi
\]
where $\psi\in H^{1}$. Since $f^{(n)}\in H^{\alpha}$ and $1/g^{2}\in H^{1}$ by
the above, we conclude that $(1/g)^{(n)}\in H^{\alpha}({\mathbb{R}})$. Hence
$1/g\in1\dotplus H^{s}$, and the proof is complete.
\end{proof}

We now have the following analogue of the Wiener--Levy theorem for the
algebra~$1\dotplus H^{s}$.

\begin{corollary}
Assume that $\Omega$ is an open subset in~$\mathbb{C}$ and that $\phi$ is a
complex-valued function that is analytic on~$\Omega$. Denote by $M_{\Omega}$
the set of all elements~$g$ of $1 \dotplus H^{s}$ such that $\overline
{\operatorname{ran}g} \subset\Omega$. Then, for every $g \in M_{\Omega}$, the
composition $\phi\circ g$ belongs to $1 \dotplus H^{s}$ and the mapping
\[
M_{\Omega}\ni g \mapsto\phi\circ g \in1 \dotplus H^{s}
\]
is locally Lipschitz continuous.
\end{corollary}

\end{document}